\documentclass[a4paper,reqno]{amsart}
\usepackage[body={13.2cm,20.5cm}]{geometry}
\usepackage{hyperref}
\usepackage[initials,backrefs]{amsrefs}
\usepackage{amssymb}
\usepackage{mathtools}
\usepackage{enumerate}
\usepackage{bookmark}
\numberwithin{equation}{section}
\theoremstyle{plain}

\newtheorem{theo}{Theorem}
\newtheorem{lemme}{Lemma}

\theoremstyle{remark}
\newtheorem*{comment*}{Comment}
\newtheorem*{remarque*}{Remark}
\theoremstyle{definition}
\newtheorem{definition}{Definition}
\newtheorem*{assk*}{(K) Kinetic energy}
\newtheorem*{assi*}{(I) Short-range interaction}
\newtheorem*{assi0*}{(I0) Bounded interaction}
\newtheorem*{assp1*}{(P1) External (random) potential energy}
\newtheorem*{assp2*}{(P2) External (random) potential energy}
\newtheorem*{dsk*}{$\dsk$}
\newtheorem*{dskone*}{$\dskone$}
\newtheorem*{dskun*}{$\dskun$}
\newtheorem*{szero*}{$\szero$}
\newtheorem*{wone*}{$\wone$}
\newtheorem*{wtwo*}{$\wtwo$}

\providecommand{\eul}{\mathrm{e}}
\DeclareMathOperator{\dist}{dist}
\newcommand{\prob}[1]{\mathbb{P}\left\{#1\right\}}
\newcommand{\BA}{\mathbf{A}}
\newcommand{\BB}{\mathbf{B}}
\newcommand{\BC}{\mathbf{C}}

\newcommand{\BG}{\mathbf{G}}
\newcommand{\BH}{\mathbf{H}}
\newcommand{\BS}{\mathbf{S}}
\newcommand{\BU}{\mathbf{U}}

\newcommand{\BW}{\mathbf{W}}
\newcommand{\BDelta}{\mathbf{\Delta}}
\newcommand{\BPsi}{\mathbf{\Psi}}
\newcommand{\DR}{\mathbb{R}}

\newcommand{\DZ}{\mathbb{Z}}
\newcommand{\DP}{\mathbb{P}}
\newcommand{\Bu}{\mathbf{u}}
\newcommand{\Bv}{\mathbf{v}}
\newcommand{\Bx}{\mathbf{x}}
\newcommand{\By}{\mathbf{y}}

\newcommand{\rB}{\mathrm{B}}

\newcommand{\rR}{\mathrm{R}}
\newcommand{\rS}{\mathrm{S}}
\newcommand{\rT}{\mathrm{T}}
\newcommand{\eps}{\epsilon}

\newcommand{\dia}{\mathrm{D}}
\newcommand{\dir}{\mathrm{D}}
\newcommand{\neu}{\mathrm{N}}
\newcommand{\nd}{\mathrm{ND}}
\newcommand{\dsk}{\mathbf{(DS.}k\mathbf{)}}
\newcommand{\dskone}{\mathbf{(DS.}k+1\mathbf{)}}
\newcommand{\dskun}{\mathrm{(DS}.k)}
\newcommand{\dszero}{\mathbf{(DS.}0\mathbf{)}}
\newcommand{\ndrons}{\mathbf{(NDRoNS)}}
\newcommand{\szero}{\mathbf{(S.}0\mathbf{)}}
\newcommand{\wone}{\mathbf{(W1)}}
\newcommand{\wtwo}{\mathbf{(W2)}}
\begin{document}
\title[On two-particle localization at low energies]{On two-particle Anderson localization\\at low energies}
\author[T.~Ekanga]{Tr\'esor Ekanga$^{\ast}$}
\address{$^{\ast}$%
UMR 7586, Institut de Math\'ematiques de Jussieu (IMJ) et Centre National de la recherche Scientifique (CNRS), Immeuble Chevaleret
Universit\'e Paris Diderot Paris 7,
175 rue du Chevaleret, 
75013 Paris,
France}
\email{ekanga@math.jussieu.fr}
\subjclass[2010]{Primary}
\keywords{}
\date{\today}
\begin{abstract}
We prove exponential spectral localization in a two-particle lattice Anderson model, with a short-range interaction and external random i.i.d.\ potential, at sufficiently low energies. The proof is based on the multi-particle multi-scale analysis developed earlier by Chulaevsky and Suhov \cite{CS09} in the case of high disorder. Our method applies to a larger class of random potentials than in Aizenman and Warzel \cite{AW09} where dynamical localization was proved with the help of the fractional moment method.
\end{abstract}
\maketitle
\section{Introduction. Main result}

Consider the lattice $\DZ^d\times\DZ^d\cong\DZ^{2d}$, $d\geq 1$. We denote $\mathbb{D}=\left\{\Bx\in\DZ^{2d}:\Bx=(x,x)\right\}$ and $\left[[a,b]\right]:=[a,b]\cap\DZ$. Vectors $\Bx=(x_1,x_2)\in\DZ^d\times\DZ^d$ will be identified with configurations of two distinguishable quantum particles in $\DZ^d$. We denote by $\left|\,\cdot\,\right|$ the max-norm $\left\|\,\cdot\,\right\|_{\infty}$, namely for $x=(x^1,\dots,x^{2d})\in\DZ^{2d}$, 
\[
|x|=\max_{1\leq i\leq 2d}|x^i|
\]
and 
\[
|x|_1=\sum_{i=1}^{2d}|x^i|.
\]

\subsection*{The two-particle model}
We study a system of two interacting lattice quantum particles in a disordered environment, described by a random Hamiltonian $\BH_{V,\BU}(\omega)$, acting in the Hilbert space $\ell^2(\DZ^{2d})$, of the form
\begin{equation}\label{eq:def,H}
\BH_{V,\BU}(\omega)=-\Delta+\sum_{j=1,2}V(x_j,\omega)+\BU,
\end{equation}
where $\Bx=(x_1,x_2)\in\DZ^d\times\DZ^d$, $\Delta$ is the nearest-neighbor laplacian on $\DZ^{2d}$,
\begin{equation}\label{eq:def,Delta}
\Delta\Psi(\Bx)=\sum_{\substack{\By\in\DZ^{2d}\\
|\By|_1=1}}\left(\Psi(\Bx+\By)-\Psi(\Bx)\right),\quad 
\end{equation}
$V\colon\DZ^d\times\Omega\to\DR$ is a random field with i.i.d.\ (independent and identically distributed) values on $\DZ^d$, relative to some probability space $(\Omega,\mathcal{F},\mathbb{P})$, and $\BU$ is the multiplication operator by a function $\BU(\Bx)=\BU(x_1,x_2)$ which we assume bounded (but not necessarily symmetric).

Aizenmann and Warzel \cite{AW09} proved by the fractional moment method --- introduced in \cite{AM93} for single-particle systems --- the spectral and dynamical localization at low energies for such Hamiltonians under the assumption that the marginal probability distribution of the i.i.d.\ random field $V$  admits a bounded probability density $\rho_V$, satisfying some additional conditions. 

In this paper, using Multi-Scale Analysis (MSA) as in \cite{CS09}, we prove exponential localization at low energies under the much weaker assumption of log-H\"older continuity of the marginal distribution function $F_V$ of the field $V$. 

\subsubsection*{Assumption on $V$}
Specifically, we require that for some $\beta\in (0,1)$, some large enough $q_0>0$ and all sufficiently large $L>0$,
\begin{equation}\label{eq:cond.V}
\sup_{a\in \DR}\mathbb{P}\left\{V(0,\omega)\in\left[a,a+\eul^{-L^\beta}\right]\right\}\leq L^{-q_0}.
\end{equation}

\subsubsection*{Assumptions on $\BU$}
The interaction potential $\BU$ is assumed to be bounded, non-negative and to satisfy the following short-range condition:
\begin{align}\label{eq:cond.U}
&\text{There exists }0\leq r_0<+\infty\text{ such that}\notag\\
&\left|x_1-x_2\right|>r_0\implies\BU(x_1,x_2)=0.
\end{align}

\begin{remarque*}
The assumption of non-negativity of the interaction potential is not essential for our main result (Theorem \ref{thm:main.result}) on Anderson localization for two-particle systems. However, it allows to simplify the adaptation of the two-particle MSA scheme proposed in \cite{CS09} to the case of weak disorder at low energies. We plan to address a more general class of interacting $N$-particle Anderson models at low energies, with any $N\ge 2$, in a separate paper.
\end{remarque*}

We denote by $\sigma\left(\BH\left(\omega\right)\right)$ the spectrum of $\BH(\omega)$. It follows from our assumptions and from well-known results that the quantity 
\[
E^0:=\inf \sigma(\BH(\omega))
\]
is non-random, although it may be infinite, e.g., for gaussian random potentials.

Given an arbitrary finite lattice cube 
\[
\BC_L(\Bu):=\left\{\Bx\in \DZ^{2d}:\ |\Bx-\Bu|\leq L\right\}
\]
we will consider the finite-volume approximation $\BH_{\BC_L(\Bu)}$ of $\BH$ defined by
\begin{align}\label{eq:restriction,H}
\BH_{\BC_L(\Bu)}&:=\BH_{\BC_L(\Bu)}^{\dir}\\
&:=\BH\vert_{\ell^2(\BC_L(\Bu))}
\text{ with Dirichlet boundary conditions on } \partial^+ \BC_L(\Bu),\notag
\end{align}
where the boundary $\partial^+\BC_L(\Bu)$ is
\begin{equation} \label{eq:def,boundary}
\partial^+\BC_L(\Bu)=\left\{\Bv\in \DZ^{2d}\setminus\BC_L(\Bu)\mid\dist\bigl(\Bv,\BC_L(\Bu))=1\right\}.
\end{equation}
Our main result is

\begin{theo}[localization at low energies] \label{thm:main.result}
Let $\BH_{V,\BU}(\omega)$ be the random hamiltonian defined in \eqref{eq:def,H}. Suppose that $V$ is an i.i.d.\ random field satisfying \eqref{eq:cond.V}, and that the potential interaction $\BU$ is bounded, non-negative and satisfies \eqref{eq:cond.U}. Let $E^0=\inf\sigma(\BH)$. 

Then there exists $E^*>E^0$ such that 
\begin{enumerate}[\rm(i)]
\item
the spectrum of $\BH(\omega)$ in $[E^0,E^*]$ is pure point, 
\item
all its eigenfunctions $\Psi_n(\omega)$ with eigenvalues $E_n(\omega)\in[E^0,E^*]$ are exponentially decaying at infinity with a positive non-random rate of decay $m>0$:
\begin{equation}\label{eq:cond,eigenfunctions}
\left|\Psi_n(\Bx)\right|\leq C_n(\omega)\eul^{-m\left|\Bx\right|}.
\end{equation}
\end{enumerate}
\end{theo}

To prove Theorem \ref{thm:main.result}, we use an adaptation of the MSA to the two-particle interacting systems, following \cite{CS09}. Given a finite cube $\BC_L(\Bu)\subset\DZ^{2d}$, introduce the resolvent of the operator $\BH_{\BC_L(\Bu)}$,
\begin{equation} \label{eq:def.resolvent}
\BG_{\BC_L(\Bu)}(E):=\left(\BH_{\BC_L(\Bu)}-E\right)^{-1},\quad E\in\DR\setminus\sigma(\BH_{\BC_L(\Bu)}).
\end{equation}
Its matrix elements $\BG_{\BC_L(\Bu)}(\Bx,\By;E)$ in the canonical basis $\{\delta_{\Bx}\}$ in $\ell^2(\DZ^{2d})$ are usually called the Green functions of the operator $\BH_{\BC_L(\Bu)}$:
\begin{equation}
\BG_{\BC_L(\Bu)}(\Bx,\By;E)
=\left\langle\left(\BH_{\BC_L(\Bu)}-E\right)^{-1}\delta_{\Bx},\delta_{\By}\right\rangle,\ \Bx,\By\in \BC_L(\Bu).
\end{equation}

According to the general MSA approach, the exponential localization will be derived from Theorem \ref{thm:DS.k} below. To formulate it, we introduce the following notion.

\begin{definition}[$(E,m)$-singular]\label{def:NS}
Let $m>0$ and $E\in\DR$. A cube $\BC_L(\Bu)\subset\DZ^{2d}$ is called $(E,m)$-non-singular ($(E,m)$-NS) if
\begin{equation}\label{eq:def,GFs}
\max_{\Bv\in \partial^-\BC_L(\Bu)}\left|\BG_{\BC_L(\Bu)}(\Bu,\Bv;E)\right|\leq\eul^{-mL}.
\end{equation}
Otherwise, it is called $(E,m)$-singular ($(E,m)$-S).
\end{definition}

Let $\BS$ be the symmetry $\Bx=(x_1,x_2)\mapsto\BS\Bx=(x_2,x_1)$ in the lattice $\DZ^{2d}=\DZ^d\times\DZ^d$ (with $x_1, x_2\in\DZ^d$). The ``symmetrized distance'' is defined in $\DZ^{2d}$ by
\begin{equation}\label{eq;def,dS}
d_{\BS}(\Bx,\By)=\min\left\{\left|\Bx-\By\right|, \left|\BS(\Bx)-\By\right|\right\}.
\end{equation}

\begin{definition}[$\ell$-distant]\label{def:distant}
Two subsets $\BA,\,\BB\subset\DZ^{2d}$ are called $\ell$-distant if 
\[
d_{\BS}(\BA,\BB)>8\ell.
\]
\end{definition}

The multiscale analysis is based on a length scale $\{L_k\}_{k\geq 0}$ which is chosen as follows. 

\begin{definition}[length scale]\label{def:length.scale}
The length-scale $\{L_k\}_{k\geq 0}$ is a sequence of integers defined by the initial length-scale $L_0>2$, and by the recurrence relation $L_{k+1}=\lfloor L_k^{\alpha}\rfloor$, $k\geq 0$ where $1<\alpha<2$ is some fixed number. In this paper, $\alpha=3/2$.
\end{definition}

The length scale $\{L_k\}_{k\geq 0}$ is assumed to be chosen at the beginning of the multiscale analysis, except that in the course of the analysis it is often required that $L_0$ be large enough.

\begin{definition} \label{def:sequence.mk}
Given a positive number $m_0>0$, we define a positive sequence $m_k$ depending upon a positive number $\gamma>0$
\begin{equation}\label{eq:sequence.mk}
m_k=m_0\prod_{j=1}^k (1-\gamma L_j^{-1/2}), \quad k\geq 1.
\end{equation}
\end{definition}

It will be assumed that $L_0$ is large enough so that 

\[
\prod_{j=1}^{\infty}(1-\gamma L_j^{-1/2})\geq \frac{1}{2}
\]

We introduce the following property of pairs of two-particle cubes of size $L_k$:

\begin{dsk*}
For any pair of $L_k$-distant cubes $\BC_{L_k}(\Bu)$ and $\BC_{L_k}(\Bv)$:
\begin{equation}\label{eq:def,DS.k}
\mathbb{P}\left\{\exists\,E\in I:\BC_{L_k}(\Bu)\text{ and }\BC_{L_k}(\Bv) \text{ are } (E,m_k)\text{-S}\right\}\leq L_k^{-2p},
\end{equation}
where $p>12d$, and $I=[E^0,E^*]$ with $E^*>E^0$, are fixed.
\end{dsk*}

\begin{comment*}
This property depends on $m_k$, $p$, $L_0$, and $E^*$. Therefore, it would be better to use a more precise notation, like $\dsk_{m_k,p,L_0,E^*}$ or even simply $\dsk_{m_k}$. 
\end{comment*}


The analogous property for one-particle cubes in $\DZ^d$ is as follows:

\begin{dskun*}
For any pair of disjoint cubes $C_{L_k}(u)$ and $C_{L_k}(v)$:
\begin{equation}\label{eq:def,DS.k}
\mathbb{P}\left\{\exists\,E\in I:C_{L_k}(u)\text{ and }C_{L_k}(v) \text{ are } (E,m_k)\text{-S}\right\}\leq L_k^{-2\tilde p},
\end{equation}
where $\tilde p>2d$, and $I=[E^0,E^*]$ with $E^*>E^0$, are fixed.
\end{dskun*}

For a single-particle random Hamilonian of the form $H=-\Delta+V(x,\omega)$ in $\ell^2(\DZ^d)$ we have the well known result:

\begin{theo}[one-particle estimate]\label{thm:DS.kun}
Let $\tilde p>2d$ be fixed. Then, provided $L_0$ is large enough, there exists $E^*=E^*(\tilde p)>E^0$  such that $\dskun_{m_k,p,L_0,E^*}$ holds true for all $k\geq 0$.
\end{theo}

\begin{proof}
See, e.g., \cite{CS09}*{Result 9.8 \& Chapters 10-11}.
\end{proof}

Here we will prove the same result for two-particle random Hamiltonian and we will be allowed to use Theorem \ref{thm:DS.kun} in which we assume that the exponent $\tilde{p}$ satisfies:
\[
\tilde{p}>\frac{9}{4}p+\frac{15}{2}d
\]

\begin{theo}[two-particle estimate]\label{thm:DS.k}
Let $p>12d$ be fixed. Then, provided $L_0$ is large enough, there exist $E^*=E^*(p)>E^0$ such that $\dsk_{m_k,p,L_0,E^*}$ holds true for all $k\geq 0$.
\end{theo}

The proof is based on induction in $k$. Note that the initial length-scale estimate (for $L_0$ sufficiently large) uses the Combes--Thomas estimate and the Lifshitz tails phenomenon, essentially in the same way as for single-particle models \cites{St01,Kir08}. In fact, the single- or multi-particle structure of the potential energy is not crucial for such a bound. The inductive step is performed almost in the same way as in the case of high disorder (see \cite{CS09}). It uses Wegner-type estimates proved in \cite{CS08} (see \cite{Weg81} for the original Wegner estimate). Note, however, that unlike the high disorder regime, the value of the mass $m>0$ may be small, depending upon the amplitude of the random potential $V$. Namely, if the random external potential has the form $gV(x;\omega)$, then the value of the mass $m=m(g)\to 0$ as $\left|g\right|\to 0$.

The derivation of the spectral localization from the bounds of the multi-particle MSA can be obtained in the same way as in the case of high disorder. The following statement is a reformulation of \cite{CS09}*{Theorem 1.2}. In turn, the main idea of the proof goes back to \cite{DK89}. In a different form, a similar argument appears already in \cite{FMSS85}.

\begin{theo} \label{thm:DS.k.implies.loc}
Suppose that $\dsk$ holds true  for some $E^*>E^0$. Then, for $\DP$-almost all $\omega$ 
\begin{enumerate}[\rm(i)]
\item
the spectrum of $\BH(\omega)$ in $(-\infty, E^*]$ is pure point, 
\item
there exists a non-random number $m>0$ such that all eigenfunctions $\Psi_n(\omega)$ of $\BH(\omega)$ with eigenvalues $E_n(\omega)\leq E^*$ are exponentially decaying at infinity with rate $m$:
\begin{equation}\label{eq:def,eigenfunction2}
\left|\Psi_n(\Bx)\right|\leq C_n(\omega)\eul^{-m\left|\Bx\right|}.
\end{equation}
\end{enumerate}
\end{theo}

\begin{proof}
See \cite{CS09}*{Theorem 1.2}.
\end{proof}

Theorem \ref{thm:main.result} derives clearly from Theorem \ref{thm:DS.k.implies.loc} and Theorem \ref{thm:DS.k}. Therefore, it only remains to prove Theorem \ref{thm:DS.k}, i.e., to check property $\dsk$ for all $k\geq 0$.

The results of this paper were announced in \cite{E11}.

\section{The two-particle MSA scheme}\label{sec:MSA,Scheme}

We now outline the two-particle MSA which is used for the proof of Theorem \ref{thm:DS.k}.  

The following definition depends on a parameter $0<\beta<1$. For our purposes, we take $\beta=1/2$, but we keep $\beta$ in all formulae to show the dependence on this parameter.

\begin{definition}[$E$-resonant]\label{def:E-R}
Let $E\in\DR$ be given. A cube $\BC_L(\Bv)\subset\DZ^{2d}$ of size $L\geq 2$ is called $E$-resonant ($E$-R) if
\begin{equation} \label{eq:def,E-R}
\dist\left[E,\sigma\left(\BH_{\BC_L(\Bv)}\right)\right]<\eul^{-L^\beta}.
\end{equation}
Otherwise it is called $E$-non-resonant ($E$-NR).
\end{definition}

The next definition depends on the parameter $\alpha>1$ which governs the length scale of our multiscale analysis. For our purposes, we take $\alpha=3/2$, but we keep $\alpha$ in all formulae.

\begin{definition}[$E$-completely non-resonant]\label{def:E-CNR}
Let $E\in\DR$ be given. A cube $\BC_L(\Bv)\subset\DZ^{2d}$ of size $L\geq 2$ is called $E$-completely non-resonant ($E$-CNR) if it does not contain any $E$-R cube of size $\geq L^{1/\alpha}$. In particular, $\BC_L(\Bv)$ is itself $E$-NR.
\end{definition}

Given $L_0>2$, we introduce the following properties $\wone$ and $\wtwo$ of the random Hamiltonians $\BH_{\BC_{\ell}}$, $\ell\geq L_0$:

\begin{wone*}
For any cube $\BC_{\ell}(\Bx)$ of size $\ell\geq L_0$ and any $E\in\DR$,
\begin{equation}\label{eq:def,W1}
\mathbb{P}\left\{\BC_{\ell}(\Bx)\ \text{is}\ \text{not}\ E\text{-CNR}\right\}<\ell^{-q},
\end{equation}
where $q>4p$ and $L_0>2$ are given.
\end{wone*}

\begin{wtwo*}
For any $\ell$-distant cubes $\BC_{\ell}(\Bx)$ and $\BC_{\ell}(\By)$ of size $\ell\geq L_0$,
\begin{equation}\label{eq:def,W2}
\mathbb{P}\left\{\exists\ E\in\DR:\text{neither}\ \BC_{\ell}(\Bx)\ \text{nor}\ \BC_{\ell}(\By)\ \text{is}\ E\text{-CNR}\right\}<\ell^{-q}
\end{equation}
where $q>4p$ and $L_0>2$ are given.
\end{wtwo*}

\begin{comment*}
These properties depend on $q$ and $L_0$. Hence, better notations would be $\wone_{q,L_0}$ and $\wtwo_{q,L_0}$.
\end{comment*}

\begin{lemme}[Wegner-type estimates]\label{lem:validity.W1.W2}
Let $q_1,\,q_2>0$ and $L_0>0$ be given. Under assumptions \eqref{eq:cond.V} on the random potential $V(x,\omega)$ and assumption \eqref{eq:cond.U} on the interaction potential $\BU$, properties $\wone$ for $q=q_1$ and $\wtwo$ for $q=q_2$ hold true for any $\ell\geq L_0$ provided $L_0$ and $q_0$ are large enough.
\end{lemme}

\begin{proof}
(i) We first prove $\wone$ for a given $q=q_1$ provided $L_0$ and $q_0$ are large enough. Let $\ell\geq L_0$, $\BC_{\ell}(\Bx)$ and $E\in\DR$ be fixed. We have:
\begin{align*}
&\mathbb{P}\left\{\BC_{\ell}(\Bx)\text{ is not $E$-CNR}\right\}\\
&\qquad\qquad
=\mathbb{P}\left\{\exists\,\By\in \BC_{\ell}(\Bx),\ \exists\,\ell':\ \ell^{1/\alpha}\leq \ell'\leq \ell,\ \BC_{\ell'}(\By)\ \text{is $E$-R}\right\}\\
&\qquad\qquad
\leq\left|\BC_{\ell}(\Bx)\right|(\ell+1)\times\mathbb{P}\left\{\BC_{\ell'}(\By)\text{ is $E$-R}\right\}\\
&\qquad\qquad
\leq (2\ell+1)^{2d}(\ell+1)\times\mathbb{P}\left\{\dist\left[E,\sigma\left(\BH_{\BC_{\ell'}(\By)}\right)\right]<\eul^{-\ell'^{\beta}}\right\}
\intertext{then by the basic one-particle Wegner estimate \cite{CS08}*{Theorem 1}}
&\qquad\qquad
\leq (2\ell+1)^{2d+1}\times\left|\BC_{\ell'}(\By)\right|\times|C_{\ell'}(y_1)|\times\sup_{a\in\DR}\mathbb{P}\left\{V(0,\omega)\in\left[a,a+2\eul^{-\ell'^{\beta}}\right]\right\}
\intertext{which gives, using assumption \eqref{eq:cond.V}, and provided $L_0$ and $q_0$ are large enough} 
&\qquad\qquad
\leq (2\ell+1)^{5d+1}\ell^{-q_0/\alpha}<\ell^{-q_1}.
\end{align*}
(ii) 
Now we prove $\wtwo$ for a given $q=q_2$ provided $L_0$ and $q_0$ are large enough. Let $\ell\geq L_0$ and $\BC_{\ell}(\Bx)$, $\BC_{\ell}(\By)$ be fixed. We have:
\begin{align*}
&\mathbb{P}\left\{\exists\,E\in\DR:\text{neither }\BC_{\ell}(\Bx)\text{ nor }\BC_{\ell}(\By)\text{ is $E$-CNR}\right\}\\
&\qquad
=\mathbb{P}\{\exists\,E\in\DR,\ \exists\,\Bu\in\BC_{\ell}(\Bx),\ \exists\,\Bv\in\BC_{\ell}(\By),\ \exists\,\ell_1,\ell_2\text{ with }\ell^{1/\alpha}\leq\ell_1,\ell_2\leq\ell:\\ 
&\qquad\qquad\quad
\BC_{\ell_1}(\Bu)\text{ and }\BC_{\ell_2}(\Bv)\text{ are $E$-R}\}\\
&\qquad
\leq \left|\BC_{\ell}(\Bx)\right|\left|\BC_{\ell}(\By)\right|(\ell+1)^2\times\mathbb{P}\left\{\exists\,E\in\DR: \BC_{\ell_1}(\Bu),\ \BC_{\ell_2}(\Bv)\text{ are $E$-R}\right\}\\
&\qquad
\leq (2\ell+1)^{4d}(\ell+1)^2\times\mathbb{P}\{\exists\,E\in\DR:\dist[E,\sigma(\BH_{\BC_{\ell_j}(\Bu)})]<\eul^{-\ell_j^{\beta}},\ j=1,2\}\\
&\qquad
\leq (2\ell+1)^{4d}(\ell+1)^2\times\mathbb{P}\{\exists\,E\in\DR:\dist[E,\sigma(\BH_{\BC_{\ell_j}(\Bu)})]<\eul^{-(\ell_1\wedge\ell_2)^{\beta}},\ j=1,2\} \\
&\qquad
\leq(2\ell+1)^{4d}(\ell+1)^2\times\mathbb{P}\{\dist[\sigma(\BH_{\BC_{\ell_1}(\Bu)}),\sigma(\BH_{\BC_{\ell_2}(\Bv)})]<2\eul^{-(\ell_1\wedge\ell_2)^{\beta}}\}
\intertext{then by the basic two-particle Wegner-type estimate \cite{CS08}*{Theorem 2}}
&\qquad
\leq(2\ell+1)^{4d}(\ell+1)^2\times\left|\BC_{\ell_1}(\Bu)\right|\left|\times\BC_{\ell_2}(\Bv)\right|\times\max\left\{|C_{\ell_1}(u_1)|,|C_{\ell_2}(v_1)|\right\}\\
&\qquad\quad
\times\sup_{a\in\DR}\mathbb{P}\left\{V(0,\omega)\in \left[a,a+4\eul^{-(\ell_1\wedge\ell_2)^{\beta}}\right]\right\}
\intertext{which gives, using assumption \eqref{eq:cond.V}, and provided $L_0$ and $q_0$ are large enough}
&\qquad
\leq (2\ell+1)^{9d+2}\ell^{-q_0/\alpha}<\ell^{-q_2}.\qedhere
\end{align*}
\end{proof}

We consider now a property which serves as replacement of $\dszero$.

\begin{szero*}
For any cube $\BC_{L_0}(\Bx)\subset\DZ^{2d}$,
\begin{equation}\label{eq:def,S.0}
\mathbb{P}\left\{\exists\,E\in[E^0,E^*]:\BC_{L_0}(\Bx)\text{ is $(E,m_0)$-S}\right\}<L_0^{-2p},
\end{equation}
where $E^*>E^0$, $m_0>0$ and $L_0\geq 2$ are given.
\end{szero*}

Obviously, property $\szero$ implies property $\dszero$, so we focus on the former. Property $\szero$ is proven in \cite{CS09} in the case of high disorder. Our proof presented here is completely different. It uses the Combes--Thomas estimate and the well-known ``Lifshitz tails'' phenomenon.

We start with the Combes--Thomas estimate, formulated for fairly general discrete Schr\"odinger operators $H=-\Delta+W$, acting in a finite-dimensional Hilbert space $\ell^2(\Lambda)$, $\Lambda\subset\DZ^{n}$ finite, where $\Delta$ is the nearest-neighbor discrete Laplacian. It is deterministic, and the structure of the potential $W(x)$ is irrelevant. This allows to apply it to the two-particle Hamiltonian $\BH=-\BDelta+\BW$, acting in the space $\ell^2( \BC_L(\Bu))$, with $\BC_L(\Bu)\subset\DZ^{2d}$ and $\BW(\Bx)=V(x_1;\omega)+V(x_2;\omega)+\BU(\Bx)$.

\begin{lemme}[Combes--Thomas estimate]\label{lem:CT1}
Let $\BH\colon\ell^2(\Lambda)\to\ell^2(\Lambda)$, $\Lambda\subset\DZ^n$. Suppose that $E\in\DR$ satisfies $\dist(E,\sigma(\BH))=\delta\leq 1$. Then, for any $x,y\in\Lambda$,
\begin{equation}
|(\BH - E)^{-1}(x,y)| \le \frac{2}{\delta}\,\eul^{ - \frac{\delta}{12\, n}|x-y|_1}
 \le \frac{2}{\delta}\,\eul^{ - \frac{\delta}{12\, n}|x-y|}.
\end{equation}
\end{lemme}

\begin{proof}
See \cite{Kir08}*{Theorem 11.2}.
\end{proof}

Next, we need the following statement, also applying to a general Schr\"odinger operator on finite subsets of a lattice of arbitrary dimension $n$. It summarizes well-known results (cf.~\cites{Kir08,St01} and references therein) from the spectral theory of a random one-particle Schr\"odinger operator $H(\omega)=-\Delta+V(x;\omega)$.

\begin{lemme}\label{lem:CT}
Let $H(\omega)=-\Delta+V(x;\omega)$ be a random Schr\"odinger operator, with non-negative i.i.d.\ random potential $V(\,\cdot\,;\omega)$, restricted to a cube $C_\ell(u)\subset\DZ^d$ with Dirichlet boundary conditions. Suppose the random variables $V(x;\omega)$ are non-constant and nonnegative.
Then for any  $C>0$ and arbitrary large $L_0>0$, there exists $c>0$ such that the lowest eigenvalue $E_0(\omega)$ of $H_{C_{L_0}(u)}(\omega)$ satisfies the bound
\begin{equation}
\prob{E_0(\omega) \le  2C L_0^{-1/2} } \leq \eul^{-c|C_{L_0}(u)|^{1/4}}.
\end{equation}
\end{lemme}

\begin{remarque*}
This lemma is actually proven for the lowest eigenvalue $E_0^{\neu}(\omega)$ of the finite-volume hamiltonian $H^{\neu}(\omega)=-\Delta+V(x;\omega)$ with Neumann boundary conditions. By Dirichlet--Neumann bracketing, the same bound holds true for the lowest eigenvalue $E_0(\omega)=E_0^{\dir}(\omega)$ of the finite-volume hamiltonian $H_{C_{L_0}(u)}(\omega)=H_{C_{L_0}(u)}^{\dir}(\omega)$ with Dirichlet boundary conditions.
\end{remarque*}

\begin{proof}
See \cite{Kir08}*{Estimate (11.16)} which follows from the study of Lifschitz tails, precisely from \cite{Kir08}*{Estimate (6.10)} and \cite{Kir08}*{Lemma 6.4}.
\end{proof}

\begin{lemme}\label{lem:CTtwo}
Consider a random Schr\"odinger operator $\BH(\omega) = -\BDelta + \BW(\Bx;\omega)$ in a lattice cube $C_{\ell}(u)\subset\DZ^n$ with Dirichlet boundary conditions. Suppose that:
\begin{enumerate}[\rm(a)]
\item
The random external potential $V(x;\omega)$ and the interaction potential $\BU(x_1,x_2)$ are non-negative.
\item
For any $\eps>0$, $\prob{ V(x;\omega) < \eps} >0$, i.e., $0$ is the sharp lower bound for the values of the random potential $V$.
\item
The random variables $V(x;\omega)$ are non-constant: $\prob{V(x;\omega)>0}>0$.
\end{enumerate}
Then for any  $C>0$ and arbitrary large  $L_0>0$ there exist $c>0$ such that the lowest eigenvalue $E_0(\omega)$ of $\BH_{C_{L_0}(u)}(\omega)$ satisfies the bound
\begin{equation}
\prob{E_0(\omega) \le  2C L_0^{-1/2} } \leq \eul^{-c|C_{L_0}(u)|^{1/4}}.
\end{equation}
\end{lemme}

\begin{proof}
The interaction potential $\BU$ is non-negative, so that by min-max principle, the lowest eigenvalue $E_0(\omega)$  of $\BH_{C_{L_0}(u)}(\omega)$ is bounded from below by the lowest eigenvalue $E_0^{\neu}(\omega)$  of operator  $-\BDelta + V(x_1;\omega) + V(x_2;\omega)$. This latter operator can be written as follows:
\[
-\BDelta + V(x_1;\omega) + V(x_2;\omega) = H^{(1)} \otimes I^{(2)}
+ I^{(1)} \otimes H^{(2)},
\]
where $H^{(j)} = -\Delta + V(x_j;\omega)$, $j=1,2$. As a result, $E_0^{\neu}(\omega)$ must have the form $E_0^{\neu}(\omega) = E^{(1)}_0(\omega) + E^{(2)}_0(\omega)$, where $E^{(j)}_0(\omega)$ is the lowest eigenvalue of $ H^{(j)}$. Finally, $E_0^{(1)}(\omega)$ and $E_0^{(2)}(\omega)$ are non-negative due to the non-negativity of the external potential, so that for any $s \ge 0$
\[
\prob{E_0^{\neu}(\omega) \le s }\le \prob{E_0^{(1)}(\omega)\le s}.
\]
Now the assertion follows from Lemma \ref{lem:CT} applied to the single-particle Schr\"odinger operator $H^{(1)}$.
\end{proof}

Lemma \ref{lem:CTtwo} leads directly to the initial scale estimate for our two-particle model.

\begin{theo}[initial scale estimate]\label{thm:initial.scale}
For any $p>0$, $C>0$ and $L_0$ large enough, there exists $E^{*}>E^0$  such that properties $\szero$ and $\dszero$ hold true for some  $m_0\ge CL_0^{-1/2}>0$.
\end{theo}

\begin{proof}
Use Lemma \ref{lem:CTtwo} and Lemma \ref{lem:CT1}.
\end{proof}

To complete the inductive step of the two-particle MSA, it only remains to prove

\begin{theo}\label{thm:inductive,MSA}
There exists $0<L^*<\infty$ such that, for any $L_0\geq L^*$ and any $k\geq 0$,
\[
\dsk_{m_k}\implies\dskone_{m_{k+1}}
\]
\end{theo}

For the proof we introduce

\begin{definition}[interactive cube]\label{def:diagonal}
Let $r_0>0$ be as in \ref{eq:cond.U} and let
\begin{equation}
\mathbb{D}_{r_0}=\left\{\Bx=(x_1,x_2)\in\DZ^{2d}: |x_1-x_2|\leq r_0\right\}.
\end{equation}
A two-particle cube $\BC_L(\Bu)$ is called interactive (I) when $\BC_L(\Bu)\cap\mathbb{D}_{r_0}\neq\varnothing$. Otherwise it is called non-interactive (NI). 
\end{definition}

\begin{remarque*}
The interaction potential $\BU$ vanish identically on any non-interactive cube.
\end{remarque*}

The procedure of deducing property $\dskone$ from $\dsk$ is done separately for the following three cases:
\begin{enumerate}[(I)]
\item
Both $\BC_{L_{k+1}}(\Bx)$ and $\BC_{L_{k+1}}(\By)$ are NI-cubes.
\item
Both $\BC_{L_{k+1}}(\Bx)$ and $\BC_{L_{k+1}}(\By)$ are I-cubes.
\item
One of the cubes is I, while the other is NI.
\end{enumerate}
More precisely:
\begin{enumerate}[(i)]
\item
In Section \ref{sec:non-interactive.cubes} we prove $\dsk^{\text{(I)}}$ for any $k\geq 0$.
\item
In Section \ref{sec:interactive.cubes} we prove $\dsk^{\text{(I,II)}}\implies\dskone^{\text{(II)}}$.
\item
In Section \ref{sec:mixed.cubes} we prove $\dsk^{\text{(I,II)}}\implies\dskone^{\text{(III)}}$. 
\end{enumerate}
All cases require the use of property $\wone$ and/or $\wtwo$.

\section{Case (I). Non-interactive pair of singular cubes}\label{sec:non-interactive.cubes}

In this section, we aim to prove $\dsk$ for any $k\geq 0$ and any pair of $L_k$-distant non-interactive cubes $\BC_{L_k}(\Bx)$ and $\BC_{L_k}(\By)$. In that case the interaction vanishes and we are mainly reduced to the one-particle case.

Let $\BC_{L_k}(\Bu)\subset\DZ^{2d}$ be a non-interactive cube, $\Bu=(u_1,u_2)$:
\begin{equation} \label{eq:cartesian,product}
\BC_{L_k}(\Bu)=C_{L_k}(u_1)\times C_{L_k}(u_2).
\end{equation}
Since $\BU$ vanishes on the non-interactive cube $\BC_{L_k}(\Bu)$, we have, for $\Bx=(x_1,x_2)\in \BC_{L_k}(\Bu)$,
\begin{equation}\label{eq:decomp,H}
(\BH_{\BC_{L_k}(\Bu)}\BPsi)(\Bx)=\sum_{|\By|_1=1}\BPsi(\Bx+\By)+\left(V(x_1,\omega)+V(x_2,\omega)\right)\BPsi(\Bx),
\end{equation}
which can be written (take $\BPsi=\Psi_1\otimes\Psi_2$)
\begin{equation}\label{eq; decompalg,H}
\BH_{\BC_{L_k}(\Bu)}=H^{(1)}_{C_{L_k}(u_1)}\otimes I^{(2)}+I^{(1)}\otimes H^{(2)}_{C_{L_k}(u_2)}.
\end{equation}
Here $H^{(j)}_{C_{L_k}(u_j)}$ is the single-particle Hamiltonian acting on $\Psi_j$, $x_j\in C_{L_k}(u_j)$, $j=1,2$:
\begin{equation}\label{eq:def,single,H}
\left(H^{(j)}_{C_{L_k}(u_j)}\Psi_j\right)(x_j)=\sum_{\substack{y_j\in C_{L_k}(u_j)\\ |y_j|=1}}\Psi_j(x_j+y_j)+V(x_j,\omega)\Psi_j(x_j)
\end{equation}
and $I^{(\hat j)}$ is the identity operator on the complementary space.

In the proof we are using the validity of the bound $\dskun$ for one-particle random Schr\"odinger operators like $H^{(j)}$ at low energies, provided $E^*$ is sufficiently close to $E_0$:

\begin{definition}[$m$-tunnelling]\label{def:tunnelling}
Let $I=[E^0,E^*]$ with $E^*>E^0$ and $m>0$ be fixed. 
\begin{enumerate}[(i)]
\item
A single-particle cube $C_{L_k}(u)\subset\DZ^d$ is called $m$-tunnelling ($m$-T) if there exists $E\in I$ and two disjoint cubes $C_{L_{k-1}}(v_1)$, $C_{L_{k-1}}(v_2)\subset C_{L_k}(u)$ which are $(E,m)$-S with respect to an Hamiltonian like $H^{(j)}$, $j=1,2$. Otherwise it is called $m$-non tunnelling ($m$-NT).
\item
A two-particle non-interavtive cube $\BC_{L_k}(\Bu)=C_{L_k}(u_1)\times C_{L_k}(u_2)\subset\DZ^{2d}$ is called $m$-non-tunnelling if the single-particle cubes $C_{L_k}(u_1)$ and $C_{L_k}(u_2)$ are $m$-NT with respect to $H^{(1)}$ and $H^{(2)}$, respectively. Otherwise, it is called $m$-tunnelling.
\end{enumerate}
\end{definition}

The following statement gives a formal description of a property of NI-cubes which will be refer to as $\ndrons$ (``Non-interactive cubes are Resonant or Non-Singular'').

\begin{lemme}[\cite{CS09}*{Lemma 3.2}]\label{lem:NDRoNS}
Let $\BC_{L_k}(\Bu)=C_{L_k}(u_1)\times C_{L_k}(u_2)\subset\DZ^{2d}$ be a two-particle cube such that
\begin{enumerate}[\rm(i)]
\item
$|u_1-u_2|>2L_k+r_0$,
\item
$\BC_{L_k}(\Bu)$ is $m'$-NT for some given $m'>0$, 
\item
$\BC_{L_k}(\Bu)$ is $E$-CNR for some $E\in\DR$. 
\end{enumerate}
Then $\BC_{L_k}(\Bu)$ is $(E,m)$-NS, with
\begin{equation}     \label{eq:NDRoNS}
m=m'-L_k^{-1}\ln (2L_k+1)^d.
\end{equation}
In particular, if $L_k^{-1}\ln (2L_k+1)^d\leq \frac{m'}{2}$, (which is true for sufficiently large $L_0$), then $m\geq \frac{m'}{2}$.
\end{lemme}

\begin{proof}
See \cite{CS09}*{Lemma 3.2}. This property is established by combining known results from the single-particle localisation theory established via MSA \cite{FMSS85} or FMM \cite{AM93}. 
\end{proof}

\begin{theo}\label{thm:DSk.ND-cubes}
Let $p>0$ be fixed. There exist, $L_1^*<\infty$ and such that for any $L_0\geq L_1^*$ and any $k\geq 0$, the estimate $\dsk_{m_k,p,L_0,E^*}$ holds true for any pair of $L_k$-distant non interactive cubes of size $L_k$.
\end{theo}

\begin{proof}
Let $I=[E^0,E^*]$. We already prove $\dszero$ for some $E^*>E^0$ and some $m_0>0$. Let $k\geq 1$. Let $\BC_{L_k}(\Bx)$ and $\BC_{L_{k}}(\By)$ be two non-interactive $L_k$-distant cubes. We consider the events
\begin{align*}
\rB_{k}&=\left\{\exists\,E\in I: \BC_{L_k}(\Bx),\ \BC_{L_k}(\By)\text{ are both } (E,m_k)\text{-S}\right\},\\
\rR & =\left\{\exists\,E\in I:\text{neither }\BC_{L_k}(\Bx)\text{ nor }\BC_{L_k}(\By) \text{ is $E$-CNR}\right\},\\
\rT_{\Bx} & =\left\{\BC_{L_k}(\Bx) \text{ is $2m_k$-T}\right\}, \\
\rT_{\By} & =\left\{\BC_{L_k}(\By) \text{ is $2m_k$-T}\right\}
\end{align*}

Let $\omega\in\rB\setminus\rR$, then $\forall E\in I$, $\BC_{L_k}(\Bx)$, or $\BC_{L_k}(\By)$ is $E$-CNR. If $\BC_{L_k}(\By)$ is $E$-CNR, then it must be $2m_k$-T: otherwise, it would have been $(E,m_k)$-NS by Lemma~\ref{lem:NDRoNS}. Similarly, if $\BC_{L_k}(\Bx)$ is $E$-CNR, then it must be $2m_k$-T. This implies that
\begin{equation}\label{eq:def,Bk}
\rB\subset\rR\cup \rT_{\Bx}\cup\rT_{\By}.
\end{equation}
We estimate $\prob{\rR}$ using property $\wtwo_q$ which holds true by Lemma \ref{lem:validity.W1.W2}:
\[
\prob{\rR}\leq L_k^{-q}.
\]
We estimate $\prob{\rT_{\Bx}}$, and similarly $\prob{\rT_{\By}}$, by using Theorem \ref{thm:DS.kun} 
\[
\prob{\rT_{\Bx}}\leq \frac{(2L_k+1)}{2}^{4d}L_{k-1}^{-2\tilde{p}}
\]
\begin{align*}
\prob{\rB_k} 
& \leq \prob{\rR}+\prob{\rT_{\Bx}}+\prob{\rT_{\By}}\\
&\leq L_k^{-q} + 2 \frac{(2L_k+1)}{2}^{4d}L_{k-1}^{-2\tilde{p}}\\
&\leq L_k^{-4p} + C(d) L_k^{4d-\frac{2\tilde{p}}{\alpha}} \qquad \notag{\text{(since $q>4p$)}}\\
&\leq \frac{1}{2}L_k^{-2p} +\frac{1}{2}L_k^{-2p}
\end{align*}
since $\tilde p>\frac{9}{4}p+\frac{15}{2}\,d$ by assumption.
\end{proof}

We end this section with a lemma on non-inteactive cubes which will be useful in the next sections.

\begin{lemme}\label{lem:property.ND.cubes}
Let $\BC_{L_{k+1}}(\Bu)$ be a two-particle cube of size $L_{k+1}$. For $E\in\DR$, we denote by $M_{\nd}(\BC_{L_{k+1}}(\Bu),E)$ the maximal number of pairwise $L_k$-distant, non-interactive $(E,m_k)$-S cubes $\BC_{L_{k}}(\Bu^{(j)})\subset\BC_{L_{k+1}}(\Bu)$. Then,
\begin{equation}
\prob{\exists\,E\in I: M_{\nd}(\BC_{L_{k+1}}(\Bu),E)\geq 2}
\leq \frac{(2L_{k+1}+1)}{2}^{4d}(L_k^q+C(d)L_k^{4d-\frac{2\tilde{p}}{\alpha}}).
\end{equation}
\end{lemme}

\begin{proof}
The total number of possible pairs of centres $\Bu^{(1)}$, $\Bu^{(2)}$ is bounded by $\frac{1}{2}(2L_{k+1}+1)^{4d}$ while for a given pair of centres one can apply the probabilistic bound, i.e. $\prob{\rB_k}$ for a pair of $L_k$-distant non-interactive cubes of side length $2L_k$.
\end{proof}

\section{Case (II). Interactive pairs of singular cubes}\label{sec:interactive.cubes}

We assume here that  for a pair of $L_k$-distant two-particle cubes $\BC_{L_k}(\Bx)$ and $\BC_{L_k}(\By)$ we have that $\mathbb{P}(\rB_k)\leq L_k^{-2p}$. 

Before we proceed further, let us state a geometric assertion borrowed from \cite{CS09}. Given a two-particle cube $\BC_L(\Bu)=C_L(u_1)\times C_L(u_2)\subset\DZ^d\times\DZ^d$ we denote
\begin{equation}\label{eq:projection,2-cube}
\varPi\BC_L(\Bu)=C_L(u_1)\cup C_L(u_2)\subset\DZ^d
\end{equation}
the union of the projections of $\BC_L(\Bu)$ on the two factors of the product $\DZ^d\times\DZ^d$.

\begin{lemme}[\cite{CS09}*{Lemma 4.1}]\label{lem:disjointness.projections.2-cube}
Let $L>r_0$. If $\BC_{L}(\Bu)$ and $\BC_{L}(\Bv)$ are two interactive $L$-distant two-particle cubes, then
\begin{equation}
\varPi\BC_L(\Bu)\cap\varPi\BC_L(\Bv)=\varnothing.
\end{equation}
\end{lemme}

\begin{proof}
See \cite{CS09}*{Lemma 4.1}.
\end{proof}

Lemma \ref{lem:disjointness.projections.2-cube} is used in the proof of Lemma \ref{lem:property.D.cubes} which, in turn, is important in establishing the inductive step for a pair of distant interactive cubes.

\begin{lemme}\label{lem:property.D.cubes}
Assume that $\dsk$ holds true for all pairs of $L_k$-distant interactive cubes. Consider a two-particle cube $\BC_{L_{k+1}}(\Bu)$ and denote by $M_{\dia}(\BC_{L_{k+1}}(\Bu),E)$  the maximal number of $(E,m_k)$-S, pairwise $L_k$-distant interactive cubes $\BC_{L_{k}}(\Bu^j)\subset \BC_{L_{k+1}}(\Bu)$. Then for all $n\geq 1$,
\begin{equation}
\mathbb{P}\left\{\exists\,E\in I: M_{\dia}(\BC_{L_{k+1}}(\Bu),E)\geq 2n\right\}
\leq C(n,d) L_{k}^{4nd\alpha} L_{k}^{-2np}.
\end{equation}
\end{lemme}

\begin{proof}
Suppose that there exist interactive cubes $\BC_{L_k}(\Bu^j)\subset \BC_{L_{k+1}}(\Bx)$,
$1 \le j \le 2n$, such that any two of them are $L_k$-distant. By Lemma \ref{lem:disjointness.projections.2-cube}, for any pair $\BC_{L_k}(\Bu^{2i-1})$,  $\BC_{L_k}(\Bu^{2i})$, the respective random operators $\BH_{\BC_{L_k}(\Bu^{2i-1})}$ and $\BH_{\BC_{L_k}(\Bu^{2i})}$ are independent and so are their spectra and Green functions. Moreover the pairs of operators 
\[
\left(\BH_{\BC_{L_k}(\Bu^{2i-1})}(\omega),\BH_{\BC_{L_k}(\Bu^{2i})}(\omega)\right),\ \ i=1,\dots,n,
\] 
form an independent family. The operator $\BH_{\BC_{L_k}(\Bu^i)}$, with $i=1,\dots,2n$ is indeed measurable relative to the sigma-algebra $\mathcal{B}_i$ generated by the random variables 
\[
\left\{V(x,\omega),x\in \varPi \BC_{L_k}(\Bu^i)\right\}.
\] 
Now by Lemma \ref{lem:disjointness.projections.2-cube}, the sets $\varPi\BC_{L_k}(\Bu^i)$, $i=1,\dots,2n$, are pairwise disjoint, so that all sigma-algebras $\mathcal{B}_i$, $i=1,\dots,2n$ are independent. Thus, any collection of events $\mathrm{A}_1,\dots,\mathrm{A}_{n}$ relative to the corresponding pairs
\[
\left(\BH_{\BC_{L_k}(\Bu^{2i-1})}(\omega),\BH_{\BC_{L_k}(\Bu^{2i})}(\omega)\right),\ \ i=1,\dots,n,
\]
also form an independent family. For $i=1,\dots,n$, set:
\begin{equation*}
\mathrm{A}_i=\left\{\exists\,E\in I:\text{ both }\ \BC_{L_k}(\Bu^{2i-1}) \text{ and } \BC_{L_k}(\Bu^{2i+2})\ \text{ are }\ (E,m)\text{-S}\right\}.
\end{equation*}
Then by virtue of the inductive assumption,
\begin{equation}\label{eq:events,Aj}
\mathbb{P}(\mathrm{A}_i)\leq L_k^{-2p},\quad 1\leq i\leq n,
\end{equation}\label{eq:ind.,events,Aj}
and owing to independence  of events $\mathrm{A}_1,\dots,\mathrm{A}_{n}$, we obtain
\begin{equation} \label{eq:prob,intersect.,Aj}
\mathbb{P}\left\{\bigcap_{i=1}^{n}\mathrm{A}_i\right\}=\prod_{j=1}^{n}\mathbb{P}(\mathrm{A}_j)\leq \left(L_k^{-2p}\right)^n
\end{equation}
To complete the proof note that the total number of different families of $2n$ cubes $\BC_{L_k}\subset\BC_{L_{k+1}}(\Bx)$ with required properties is bounded from above by
\[
\frac{1}{(2n)!}\left|\BC_{L_{k+1}(\Bu)}\right|^{2n}\leq C(n,d)L_k^{4dn\alpha}
\]
\end{proof}

\begin{lemme}\label{lem:(E,J)-CNR,implies,NS}
Let $J\geq 1$ be an odd integer and $E\in\DR$. Let $\BC_{L_{k+1}}(\Bu)$ be a cube such that:
\begin{enumerate}[\rm(i)]
\item
$\BC_{L_{k+1}}(\Bu)$ is $E$-CNR,
\item
$M_{\nd}(\BC_{L_{k+1}}(\Bu),E)+M_{\dia}(\BC_{L_{k+1}}(\Bu),E)\leq J$.
\end{enumerate}
Then, for $L_0$ large enough $\BC_{L_{k+1}}(\Bu)$ is $(E,m_{k+1})$-NS for some $m'$ such that 
\[
m_{k+1}\geq m_k\biggl(1-\frac{5J+6}{(2L_k)^{1/2}}\biggr).
\]
\end{lemme}

\begin{proof}
Simple reformulation of \cite{DK89}*{Lemma 4.2}. See also \cite{CS09}*{Lemma 4.5}.  
\end{proof}

\begin{theo}\label{thm:DS.k,for,diagonal-cubes}
There exists $0<L_2^*<+\infty$ such that if $L_0\geq L_2^*$, then if for $k\geq 0$, property $\dsk$ holds true for all pairs of $L_k$-distant interactive cubes $\BC_{L_k}(\Bx)$, $\BC_{L_k}(\By)$, the $\dskone$ holds true for all pairs of $L_{k+1}$-distant interactive cubes $\BC_{L_{k+1}}(\Bx)$, $\BC_{L_{k+1}}(\By)$ 
\end{theo}

\begin{proof}
Consider a pair of $L_{k+1}$-distant two-particle interactive cubes $\BC_{L_{k+1}}(\Bx)$ and $\BC_{L_{k+1}}(\By)$. Let us set
\begin{align*}
\rB_{k+1}
&=\left\{\exists\,E\in I: \BC_{L_{k+1}}(\Bx),\text{ and }\BC_{L_{k+1}}(\By)\text{ are $(E,m_{k+1})$-S}\right\}\\
\Sigma 
&=\left\{\exists\,E\in I:\text{neither }\BC_{L_{k+1}}(\Bx)\text{ nor }\BC_{L_{k+1}}(\By)\text{ is $E$-CNR}\right\}\\
\rS_{\Bx}
&=\left\{\exists\,E\in I:M_{\nd}(\BC_{L_{k+1}}(\Bx),E)+M_{\dia}(\BC_{L_{k+1}}(\Bx),E)\geq J+1\right\}\\
\rS_{\By}
&=\left\{\exists\,E\in I:M_{\nd}(\BC_{L_{k+1}}(\By),E)+M_{\dia}(\BC_{L_{k+1}}(\By),E)\geq J+1\right\}.
\end{align*}
Let $\omega \in \rB$. Suppose that $\omega \notin \Sigma \cup \rS_{\Bx}$ so that $\forall E\in I$ either $\BC_{L_{k+1}}(\Bx)$ or $\BC_{L_{k+1}}(\By)$ is $(E,J)$-CNR and $M_{\nd}(\BC_{L_{k+1}}(\Bx),E)+M_{\dia}(\BC_{L_{k+1}}(\Bx),E)\leq J$. Thus the  cube $\BC_{L_{k+1}}(\Bx)$ cannot be $(E,J)$-CNR: indeed, by Lemma \ref{lem:(E,J)-CNR,implies,NS}, it would be $(E,m_{k+1})$-NS. So, it is the cube $\BC_{L_{k+1}}(\By)$ which is $(E,J)$-CNR and $(E,m{+1})$-S. This implies again by Lemma \ref{lem:(E,J)-CNR,implies,NS} that
\[
M_{\nd}(\BC_{L_{k+1}}(\By);E)+M_{\dia}(\BC_{L_{k+1}}(\By);E)\geq J+1.
\] 
Therefore $\omega\in \rS_{\By}$. This shows that
\begin{equation*}
\rB_{k+1}\subset \Sigma \cup \rS_{\Bx}\cup\rS_{\By}.
\end{equation*}
Therefore,
\begin{align*}
\prob{\rB_{k+1}} 
& \leq \prob{\Sigma}+\prob{\rS_{\Bx}}+\prob{\rS_{\By} }\\
&\leq  L_{k+1}^{-q}+2\prob{\rS_{\Bx}}.
\end{align*}
with $q>4p$ large enough.

It remains to estimate $\prob{\rS_{\Bx}}$. Set $J=2n+1$, then  
\[
M_{\nd}(\BC_{L_{k+1}}(\Bx);E)+ M_{\dia}(\BC_{L_{k+1}}(\Bx);E)\geq 2n+2
\]
implies that either $M_{\nd}(\BC_{L_{k+1}}(\Bx);E)\geq 2$ or $M_{\dia}(\BC_{L_{k+1}}(\Bx);E)\geq 2n$. Then by Lemma \ref{lem:property.ND.cubes} and Lemma \ref{lem:property.D.cubes}, we have
\begin{align*}
\prob{\rS_{\Bx}} 
& \leq \prob{\exists\,E\in I:M_{\nd}(\BC_{L_{k+1}}(\Bx);E)\geq 2}\\
&\quad+\prob{ \exists\,E\in I:M_{\dia}(\BC_{L_{k+1}}(\Bx);E)\geq 2n}\\
&\leq \frac{(2L_{k+1}+1)^{4d}}{2}(L_k^{-q}+C(d)L_k^{4d-\frac{-2\tilde{p}}{\alpha}})+ C(n,d)L_k^{4dn\alpha-2np} \\
&\leq C(d) \left(L_{k+1}^{4d-\frac{4p}{\alpha}} + L_{k+1}^{4d+ \frac{4d}{\alpha}-\frac{2\tilde{p}}{\alpha^2}}\right) + C(n,d)L_{k+1}^{4nd-\frac{2np}{\alpha}}  \\
&\leq C(d)\left( L_{k+1}^{4d-\frac{8p}{3}} + L_{k+1}^{-\frac{8\tilde{p}}{9}+\frac{8d}{9}+4d} + L_{k+1}^{-\frac{4p}{\alpha}+ 8d}\right) \quad \text{(by taking $n=2$ and $\alpha=3/2$)}\\
&\leq L_{k+1}^{-2p}
\end{align*}
where we used that by our assumptions $q>4p$, $p>12d$ and $\tilde{p}>\frac{9}{4}p+\frac{15}{2}d$
\end{proof}

\section{Case (III). Mixed pairs of singular cubes}\label{sec:mixed.cubes}

Now we derive property $\dskone$ in the case (III), for mixed pairs of two-particle cubes (where one cube is interactive and the other non-interactive). Here we use several properties which have been establish earlier in this paper for all scales lengths. Namely $\wone$, $\wtwo$, $\ndrons$ and the inductive assumption.

\begin{theo}\label{thm:mixed-cubes}
There exists $0<L_3^*<+\infty$ such that if $L_0\geq L_3^*$ and if for  $k\geq 0$, property $\dsk$ holds true
\begin{enumerate}
\item[\rm(i)] for every pairs of $L_k$-distant non-interactive cubes $\BC_{L_k}(\Bx)$, $\BC_{L_k}(\By)$.\\
\item[\rm(ii)] for every pairs of $L_k$-distant interactive cubes $\BC_{L_k}(\Bx)$, $\BC_{L_k}(\By)$.
\end{enumerate}

Then $\dskone$ holds true for all mixed pairs of $L_{k+1}$-distant cubes $\BC_{L_{k+1}}(\Bx)$, $\BC_{L_{k+1}}(\By)$.
\end{theo}

\begin{proof}
Consider a pair of two particles $L_{k+1}$-distant cubes $\BC_{L_{k+1}}(\Bx)$ and $\BC_{L_{k+1}}(\By)$, where $\BC_{L_{k+1}}(\Bx)$ is NI while $\BC_{L_{k+1}}(\By)$ is I. Let us set as in the previous sections
\begin{align*}
\rB_{k+1} 
&=\left\{\exists\,E\in I: \BC_{L_{k+1}}(\Bx),\ \BC_{L_{k+1}}(\By),\text{ are $(E,m_{k+1})$-S}\right\}\\
\Sigma 
& = \left\{\exists\,E\in I:\text{ neither } \BC_{L_{k+1}}(\Bx)\text{ nor } \BC_{L_{k+1}}(\By)\text{ is $E$-CNR}\right\}
\\
\rT_{\Bx} 
& = \left\{\BC_{L_{k+1}}(\Bx)\text{ is $2m_{k+1}$-T}\right\}\\
\rS_{\By} 
& = \left\{\exists\,E\in I: M_{\nd}(\BC_{L_{k+1}}(\By);E)+M_{\dia}(\BC_{L_{k+1}}(\By);E)\geq J+1\right\}.
\end{align*}
Let $\omega\in \rB_{k+1}\setminus ( \Sigma\cup \mathrm{T}_{\Bx})$, then $\forall E \in I $ either $\BC_{L_{k+1}}(\Bx)$ or $\BC_{L_{k+1}}(\By)$ is $(E,J)$-CNR  and  $\BC_{L_{k+1}}(\Bx)$ is $2m_{k+1}$-NT. By Lemma \ref{lem:NDRoNS}, $\BC_{L_{k+1}}(\Bx)$ cannot be $(E,J)$-CNR. Indeed it would have been $(E,m_{k+1})$-NS. So it is the cube $\BC_{L_{k+1}}(\By)$ which is $(E,J)$-CNR. This implies for some $E\in I$, $M_{\nd}(\BC_{L_{k+1}}(\By);E)+M_{\dia}(\BC_{L_{k+1}}(\By);E)\geq J+1$ by Lemma \ref{lem:(E,J)-CNR,implies,NS}. Therefore $\omega\in \rS_{\By}$. This shows that
\[
\rB_{k+1}\subset \Sigma \cup \mathrm{T}_{\Bx}\cup \rS_{\By},
\]
so that
\begin{align*}
\prob{\rB_{k+1}} 
&\leq \prob{\Sigma}+\prob{\rT_{\Bx}}+\prob{\rS_{\By}}\\
&\leq L_{k+1}^{-q}+\frac{1}{2}(2L_{k+1}+1)^{4d} L_{k+1}^{-\frac{2\tilde{p}}{\alpha}}+ \frac{1}{4}L_{k+1}^{-2p}\\
&\leq L_{k+1}^{-4p} + C(d)L_{k+1}^{4d-\frac{2\tilde{p}}{\alpha}}+ \frac{1}{4}L_{k+1}^{-2p}\\
&\leq \frac{1}{4}L_{k+1}^{-2p}+\frac{1}{2}L_{k+1}^{-2p}+\frac{1}{4}L_{k+1}^{-2p}\\
&\leq L_{k+1}^{-2p}\\
\end{align*}
\end{proof}

Since $\prob{\Sigma}$, $\prob{\rT_{\Bx}}$ and $\prob{\rS_{\By}}$ have already been established in the previous sections.

This completes the proof of Theorem \ref{thm:inductive,MSA}. Therefore, Theorem \ref{thm:DS.k} is also proven, since we have already proven Theorem \ref{thm:initial.scale} giving the base of induction in $k$. By virtue of Theorem \ref{thm:DS.k.implies.loc}, this completes also the proof of our main result on two-particle localization at low energies, Theorem \ref{thm:main.result}.

\subsection*{Acknowledgments} 
The author would like to thank Anne Boutet de Monvel and Victor Chulaevsky for their constant support and encouragement.

\end{document}